\documentclass[runningheads]{llncs}
\usepackage{graphicx}
\usepackage{algorithm}
\usepackage{algorithmicx}
\usepackage[noend]{algpseudocode}

\usepackage{color}

\usepackage{amsmath,amssymb,amsfonts}
\newtheorem{assumption}{Assumption}

\newcommand{\prob}{{\textbf{Min-SUKP}}}
\newcommand{\KP}{{\textbf{KP}}}
\newcommand{\SKP}{{\textbf{SKP}}}
\newcommand{\MinKP}{{\textbf{Min-KP}}}
\newcommand{\SKCC}{{\textbf{SK-CC}}}

\newcommand{\UKP}{{\textbf{UKP}}}

\begin{document}
	\title{An FPTAS for Stochastic Unbounded Min-Knapsack Problem}
	%
	%
	\author{Zhihao Jiang\inst{}\orcidID{0000-0002-9682-2476} \and
		Haoyu Zhao\inst{}\orcidID{0000-0001-6775-1421}}
	\authorrunning{Z. Jiang \and H. Zhao}
	%
	\institute{Institute for Interdisciplinary Information Sciences, Tsinghua University, China
		\email{\{jzh16,zhaohy16\}@mails.tsinghua.edu.cn}\\
	}
	\maketitle              
	\begin{abstract}
		In this paper, we study the stochastic unbounded min-knapsack problem (\prob).
		The ordinary unbounded min-knapsack problem states that:
		There are $n$ types of items, and there is an infinite number of items of each type. The items of the same type have the same cost and weight. We want to choose a set of items such that the total weight is at least $W$ and the total cost is minimized.
		The \prob~generalizes the ordinary unbounded min-knapsack problem to the stochastic setting, where the weight of each item is a random variable following a known distribution and the items of the same type follow the same weight distribution. In \prob, different types of items may have different cost and weight distributions.
		In this paper, we provide an FPTAS for \prob, i.e., the approximate value our algorithm computes is at most $(1+\epsilon)$ times the optimum, and our algorithm runs in $poly(1/\epsilon,n,\log W)$ time.
		
		\keywords{Stochastic Knapsack, Renewal Decision Problem, Approximation Algorithms}
	\end{abstract}

	\section{Introduction}
	In this paper, we study the stochastic unbounded min-knapsack problem (\prob).
	The problem is motivated by the following renewal decision problems introduced in \cite{RN29}.
	A system (e.g., a motor vehicle) must operate for $t$ units of time.
	A particular component (e.g., a battery) is essential for its operation and must be replaced each time it fails. There are $n$ different types of replacement components,
	and every kind of items has infinite supplies.
	A type $i$ replacement costs $C_i$ and has a random lifetime with distribution depending on $i$. The problem is to assign the initial component and subsequent replacements from among the types to minimize the total expected cost of providing an operative component for the $t$ units of time.
	Formally, we would like to solve the following \prob\ problem, defined as follows:
	
	\begin{problem}[stochastic unbounded min-knapsack]\label{pro-1}
		There are $n$ types of items $a_1,a_2,\dots,a_n$. For an item of type $a_i$, the cost is a deterministic value $c_i$, and the weight is random value $X_i$ which follows a known distribution $D_i$ with non-negative integer support. Let $D_i(j)$ denote $\text{Pr}\{X_i \le j\}$. Each type has infinite supplies, and the weight of each item is independent of the weight of the items of other types and other items of the same type. Besides, there is a knapsack with capacity $W$. Our objective is to insert items into the knapsack one by one until the total weight of items in the knapsack is at least $W$. The realized weight of an item is revealed to us as soon as it is inserted into the knapsack. What is the expected cost of the strategy that minimizes the expected total cost of the items we insert?
	\end{problem}
	
	\begin{remark}
		The above problem is the stochastic version of the ordinary unbounded min-knapsack problem.
		Comparing to the ordinary knapsack problem, there is an infinite number of items of each type, and the objective is to minimize the total cost
		(rather than maximize the total profit).
	\end{remark}
	
	\begin{remark}
	    It can be shown that \prob\ is NP-hard. In \cite{garey2002computers}, the authors mentioned that the unbounded knapsack problem (\UKP) is NP-hard, and it can be easily shown that the unbounded min-knapsack is NP-hard, since there is a polynomial reduction between these 2 problems. The problem \prob\ is NP-hard since it is a generalization of unbounded min-knapsack.
	\end{remark}
	
	Derman et al.\cite{RN29} discussed \prob~when the weight distributions of items are exponential and provided an exact algorithm to compute the optimal policy. Assaf \cite{RN39} discussed \prob~when the weight distributions of items have a common matrix phase type representation. 
	
	In this paper, we present a fully polynomial time approximation scheme (FPTAS) for this problem for general discrete distributions.
	
	Roughly speaking, we borrow the idea of the FPTAS for the knapsack problem
	and the method for computing the distribution of the sum of random variables \cite{RN14}.
	However, there are a few technical difficulties we need to handle.
	The outline of our algorithm is as follows.
	We first compute a constant factor approximation for the optimal cost (Section \ref{sec-a0}), and then we apply the discretization and a dynamic program based on the approximation value (Section \ref{sec-a1}).
	However, the dynamic program can only solve the problem in a restricted case where the cost for any item is `not too small' (the cost of each item is larger than a specific value). To solve the whole problem, we consider a reduction from the general setting to the restricted setting and show that the error of the reduction is negligible (Section \ref{sec-a2}).
	
	\subsection{Related Work}

	The knapsack problem is a classical problem in combinatorial optimization. The classical knapsack problem (max-knapsack problem) is the following problem: Given a set of items with sizes and costs, and a knapsack with a capacity, our goal is to select some items and maximize the total cost of selected items with the constraint that the total size of selected items does not exceed the capacity of the knapsack.
	
	The min-knapsack problem (\MinKP) \cite{RN35} is a natural variant of the ordinary knapsack problem. In the min-knapsack problem, the goal is to minimize the total cost of the selected items such that the total size of the selected items is not less than the capacity of the knapsack. Although the min-knapsack problem is similar to the max-knapsack problem, a polynomial-time approximation scheme (PTAS) for the max-knapsack problem does not directly lead to a PTAS for the min-knapsack problem. For the (deterministic) min-knapsack problem, approximation algorithms with constant factors are given in \cite{RN35,RN36,RN45}. Han and Makino \cite{RN34} considered an online version of min-knapsack, that is, the items are given one-by-one over time.
	
	There is also a line of work focusing on the FPTAS for unbounded knapsack problem(\UKP). \UKP\ is similar to the original 0-1 knapsack problem, except that there are infinite number of items of each type. The first FPTAS for \UKP\ is introduced by \cite{ibarra1975fast}, and they show an FPTAS by extending their FPTAS for 0-1 knapsack problem. Their algorithm runs in $O(n+\frac{1}{\epsilon^4}\log \frac{1}{\epsilon})$ time and needs $O(n+\frac{1}{\epsilon^3})$ space. Later, \cite{kellerer2004multidimensional} showed an FPTAS with time complexity $O(n\log n + \frac{1}{\epsilon^2}(n+\log \frac{1}{\epsilon}))$ and space complexity $O(n+\frac{1}{\epsilon^2})$. In 2018, \cite{jansen2018faster} presented an FPTAS that runs in $O(n_\frac{1}{\epsilon^2}\log^3\frac{1}{\epsilon})$ time and requires $O(n+\frac{1}{\epsilon}\log^2\frac{1}{\epsilon})$ space.
	
	However, in some applications, precisely knowing the size of each item is not realistic. In many real applications,
	we can only get the size distribution of a type of item.
	This problem leads to the stochastic knapsack problem (\SKP~\cite{ross1989stochastic}), which is a generalization of \KP. In \SKP, the cost of each item is deterministic, but the sizes of items are random variables with known distributions, and we get the realized size of an item as soon as it is inserted into the knapsack. The goal is to compute a solution policy which indicates the item we insert into the knapsack at a given remaining capacity. For the stochastic max-knapsack problem, an approximation with a constant factor was provided in the seminal work \cite{RN30}. The current best approximation ratio
	for \SKP\ is 2 \cite{Ex4,Ma2014}. An $(1+\epsilon)$ approximation with relaxed capacity (bi-criterion PTAS) is given in \cite{RN31,RN33}.
	Besides, Deshpande et al.\cite{RN47} gave a constant-factor approximation algorithm for the stochastic min-knapsack.
	
	Gupta et al.\cite{Ex2} considered a generalization of \SKP, where the cost of items may be correlated, and we can cancel an item during its execution in the policy. Cancelling an item means we can set a bounding size each time we want to insert an item, we cancel the item if the realized size of the item is larger than the bounding size. When we cancel an item, the size of the item is equal to the bounding size, and the cost of the item is zero. This generalization is referred to as Stochastic Knapsack with Correlated Rewards and Cancellations (\SKCC). Gupta et al.\cite{Ex2} gave a constant-factor approximation for \SKCC~based on LP relaxation.
	A bicriterion PTAS for \SKCC\ is provided in \cite{RN33}.
	
	\subsection{Preliminary}\label{sec-pre}

	\begin{proposition}\label{positive-support-assumption}
		Without the loss of generality, we can assume that the support of $D_i$, which is the weight distribution of an item of type $i$, has positive integer support.
	\end{proposition}
	
	We skip the proof of Proposition \ref{positive-support-assumption}. Please see the proof in Appendix \ref{proof-positive-sup}.
	
	From now on, we can suppose that each type of item has weight distribution with positive integer support.

	In \prob, the optimal item added can be determined by the remaining capacity. Let $OPT_w$ denote the expected cost of the optimal strategy when the remaining size is $w$. We can assume that the support of $D_i$ is $\{0,1,\dots,W\}$. Let $OPT_0 = OPT_{-1} = \cdots = OPT_{-W+1} = 0$. Define $d_i(j) = D_i(j) - D_i(j-1) = \text{Pr}\{X_i = j\}$. From the dynamic program, we have pseudo-polynomial time Algorithm \ref{alg-0} that can compute the exact optimal value.
	
	
	\begin{algorithm}[!ht]
		\caption{Pseudo-polynomial Time Algorithm}\label{alg-0}
		\begin{algorithmic}[1] 
			\State {$OPT_{i} \leftarrow 0$ for $-W+1\leq i\leq 0$}
			\For{$i=1\rightarrow W$}
			\State $OPT_i=\min_{j=1}^{n}\left(c_j+\sum_{k=1}^{W}d_j(k)\cdot OPT_{i-k}\right)$
			\EndFor
			\Return $OPT_W$
		\end{algorithmic}
	\end{algorithm}
	
	
	Algorithm \ref{alg-0} runs in $\textbf{poly}(n,W)$ time.
	
	In this paper, we show an FPTAS to compute $OPT_W$. Our algorithm runs in $\textbf{poly}(\frac{1}{\epsilon},n,\log W)$ time and return $OPT'_W$, which is an approximation for $OPT_W$, such that $(1-\epsilon)OPT_W \le OPT'_W \le (1+\epsilon)OPT_W$. We assume that there is an oracle $\mathcal A$ such that we can call $\mathcal A$ to get $D_i(j) = \text{Pr}\{X_i\le j\}$. Since we require that our algorithm runs in $\textbf{poly}(\frac{1}{\epsilon},n,\log W)$ time, our algorithm can call the oracle for at most $\textbf{poly}(\frac{1}{\epsilon},n,\log W)$ times.

	\section{A Constant Factor Estimation}\label{sec-a0}
	In this section, we show that there is a constant factor approximation for the optimal value. This constant factor approximation serves to estimate the optimal value roughly, and our FPTAS uses the standard discretization technique based on this rough estimation.
	
	Define $b_i=\frac{c_i}{E[X_i]}$. When we insert an item of type $i$, the expected weight is $E[X_i]$, and the cost is $c_i=b_iE[X_i]$. Suppose $m=\arg\min_i b_i$, and we will show that $2b_mW$ is a constant approximation for the optimal value $OPT_W$. Formally, we have the following lemma,
	\begin{lemma}\label{constapprox}
		For all $-W+1\le w\le W,~b_{m}w\le OPT_w \le b_{m}(w+W)$, where $m=\arg\min_i b_i$.
	\end{lemma}
	
	This lemma can be proved by induction, and please see Appendix \ref{proof-constapprox} for its formal proof.
	
	Specifically, when $w=W$, we get $b_{m}W\le OPT_W\le 2b_{m}W$ directly from the above lemma. However, when computing $b_m$, we need to enumerate the support. To avoid expensive enumeration, we can compute $E[X_i]$ approximatively. We round the realized weight $x_i$ into $2^{\lfloor\log_2{x_i}\rfloor}$. Just let
	\[ \overline{E}[X_i] = \sum_{j=1}^{W}d_i(j)2^{\lfloor\log_2{j}\rfloor} = D_i(1) + \sum_{j=0}^{\lfloor\log W\rfloor} \left( 2^{j}\cdot (D_i(2^{j+1})-D_i(2^{j})) \right). \]
	
	We have $ \frac{{E}[X_i]}{2} \leq \overline{E}[X_i] \leq {E}[X_i]$, since $\frac{x_i}{2}\leq 2^{\lfloor\log_2{x_i}\rfloor}\leq x_i$.

	Let $\overline{OPT}_W=2W\cdot \min_{i}\frac{c_i}{\overline{E}[X_i]}$.
	From the previous argument, we have $2b_{m}W\le \overline{OPT}_{W}\le 4b_{m}W$, which means $OPT_{W}\leq \overline{OPT}_{W}\leq 4OPT_{W}$.
	
	Let $T=\frac{1}{4}\overline{OPT}_{W}$. We have $\frac{1}{4}OPT_{W}\leq T\leq OPT_{W}$. $T$ is the estimation of $OPT_{W}$.
	
	\section{FPTAS Under Certain Assumption}\label{sec-a1}
	
	In this section, we discuss \prob~under the following assumption.
	
	\begin{definition}[Cheap/Expensive type]
	
	Let $\theta= \frac{\epsilon}{10n}$. We call type $i$ is an expensive type if $c_i\ge \theta T$, otherwise we call type $i$ is a cheap type.
	
	\end{definition}
	
	\begin{assumption}\label{assm}
		we assume all the types are expensive.
	\end{assumption}
	
	And we give an algorithm with approximation error at most $\epsilon T$ in this section under Assumption \ref{assm}.
	
	In general, our algorithm for \prob~is inspired from the FPTAS of the ordinary knapsack problem \cite{RN48}. We define $\hat{f}_c=\max\{w|OPT_w\le c\}$, and compute the approximation for $\hat f$. However, the support of $\hat f$ is the set of real numbers. So we discretize $\hat f$ and only compute the approximation for $\hat f_{i\delta T}$ for all $i \le \lceil\frac{1}{\delta}\rceil + 1$, where $i$ is non-negative integer and $\delta = \frac{\epsilon^2}{100n}$. In our algorithm, we use dynamic programming to compute $f_i$, which is the approximation for $\hat f_{i\delta T}$. Then we use $f_i$ to get an approximate value of $OPT_W$. Since $\hat f_{i\delta T}$ is monotonically increasing with respect to $i$, we can find the smallest $i$ such that $f_i \geq W$ and return the value $i\delta T$ as the approximate value of $OPT_W$.
	
	Now we show how to compute $f_i$. First, suppose that $\hat f_{i\delta T} = w^*$, and from the dynamic programming, we have
	\[
	OPT_{w^*} = \min_{k}\left\{c_k + \sum_{j=1}^W d_k(w^*-j)OPT_{j}\right\}.
	\]
	Since $OPT_{w}$ is non-decreasing while $w$ is increasing, recall $\hat{f}_c=\max\{w|OPT_w\le c\}$, and we get,
	\begin{align*}
	w^* =& \max\left\{w'\Bigg{|}\exists k, c_k + \sum_{j=1}^W d_k(w'-j)OPT_{j} \le i\delta T\right\}\\
	=& \max\left\{w'\Bigg{|}\exists k, c_k + \sum_{j=1}^{w'-1} d_k(w'-j)OPT_{j} \le i\delta T\right\}.
	\end{align*}
	
	Define $\hat g_w:=j\delta T$ for all $\hat f_{j-1}< w \leq \hat f_{j}$. Then $\hat g_w$ is the rounding up discretization value of $OPT_w$, and we can approximately compute $w^*$ (let $\hat w$ denote the approximate value) by
	\[\hat w=\max \left\{ w' ~\Bigg{|}~
	\exists k, \left( c_k+\sum_{j=1}^{w'-1} d_k(w'-j) \hat g_{j} \right)
	\leq i\delta T  \right\}.\]
	However, we do not have $\hat g$ during the computation. Instead, we use the following quantity to approximate $\hat g$. Given $f_0,\dots,f_{i-1}$, define $g_w:=j\delta T$ for all $f_{j-1}< w \leq f_{j}$ where $j\le i-1$, and define $g_w = i\delta T$ for all $w>f_{i-1}$. Then we have
	\begin{align}\label{g-to-f}
	f_i=\max \left\{ w' ~\Bigg{|}~
	\exists k, \left( c_k+\sum_{j=1}^{w'-1} d_k(w'-j) g_{j} \right)
	\leq i\delta T  \right\}.
	\end{align}
	
	\begin{remark}
		When we compute $f_i$, we have already gotten $f_0,f_1,\dots f_{i-1}$.
	\end{remark}
	
	To compute $f_i$, we use binary search to guess $f_i = w'$ and accept the largest $w'$ that satisfies the constraint in (\ref{g-to-f}).
	
	The pseudo-code of our algorithm is shown in Algorithm \ref{alg-1}. The detailed version of the pseudo-codes is presented in Appendix \ref{full-pseudo-code}.
	
	In details, we enumerate $i$ and compute $f_i$ until $f_i$ reaches the weight lower limit $W$. To compute $f_i$, we use binary search starting with $L=0,R=W$. In each step of binary search, let $w=(L+R)/2$ and compute $g_w$, and decide to recur in which half according to the relation between $g_w$ and $i\delta T$, until $L=R$ which means $f_i=L=R$.
	
	\begin{algorithm}[!t]
		\caption{The Dynamic Program for Computing approximate answer for \prob~under Assumption \ref{assm}}\label{alg-1}
		\begin{algorithmic}[1] 
			\State Let $\delta = \frac{\epsilon^2}{100n}$.
			\State $f_{0}\leftarrow 0$
			\For{$i=1\rightarrow \lceil\frac{1}{\delta}\rceil+1$}
			\State Compute $f_i$ using binary search according to Algorithm \ref{compute-f}
			\If{$f_i\geq W$}
			\State \Return $\hat{V}:=i\delta T$
			\EndIf
			\EndFor
		\end{algorithmic}
	\end{algorithm}
	\begin{algorithm}[!t]
		\caption{Given $w$, judge whether $f_i\geq w$ (whether $g_{w}\leq i\delta T$)}\label{compute-f}
		\begin{algorithmic}[1] 
			\For{$j=1\rightarrow n$}
			\For{$m=0\rightarrow i-2$}
			\State $P_m=Pr\left[X_j\in [w-f_{m+1},w-f_{m}) \right]$ \Comment by Binary search from oracle
			\EndFor
			\State $P_{i-1}=Pr\left[X_j\in [1,w-f_{i-1}) \right]$
			\State $g_w \leftarrow c_j+\sum_{m=0}^{i-1}P_m (m+1) \delta T$ \Comment Equation (\ref{g-to-f})
			\If{$g_w\leq i \delta T$}
			\Return true
			\EndIf
			\EndFor
			\Return false
		\end{algorithmic}
	\end{algorithm}
	
	To quantify the approximation error by algorithm \ref{alg-1}, we have the following theorem.
	
	\begin{theorem}\label{dp-err}
		
		The output $\hat{V}$ of Algorithm \ref{alg-1} satisfies $ (1-\delta)(1-\frac{\epsilon}{10})\hat{V} \leq OPT_{W} \leq \hat{V} $.
		
	\end{theorem}
	
	Generally speaking, this results can be shown in 2 steps: First, we will show that the real optimal value is upper bounded by the value computed in our algorithm, and next, we will show that under Assumption \ref{assm}, the difference between the value computed in Algorithm \ref{alg-1} and the real optimal value is upper bounded by a small value. Given these two results, we can prove Theorem \ref{dp-err}. Please see Appendix \ref{proof-dp-err} for the formal proof of Theorem \ref{dp-err}.
	
	From the above theorem, we know that the output $\hat{V}$ of Algorithm \ref{alg-1} is a $(1+\epsilon)$-approximation for $OPT_{W}$.

	\section{FPTAS in the General Case}\label{sec-a2}
	In the previous section, we show that there is an FPTAS of \prob~under Assumption \ref{assm} (when all the types are expensive). In this section, we remove Assumption \ref{assm} and show that there is an FPTAS of \prob. We will first present the general idea of our algorithm.
	
	\textbf{Our Ideas:} If we use the algorithm in the last section to compute in general case, the error will not be bounded. The key reason is that we may insert lots of items of cheap types. One idea is, we can bundle lots of items in the same cheap type $p$ into bigger items (an induced type $p'$), such that $p'$ is expensive. Then we replace type $p$ by the new type $p'$. Now, we can use the algorithm in the last section. However, we can only use bundled items even if we only want to use one item of a certain cheap item. Luckily, using some extra items of cheap items does not weaken the policy very much.

	The remaining problem is, how to compute the distribution of many items of type $p$? For example, we always use $e_p = 2^k$ items of type $p$ each time. We discretize the weight distribution $X_p$, and use doubling trick to compute the approximate distributions for $X_{p,1},\sum_{i=1}^2 X_{p,i},\sum_{i=1}^4 X_{p,i},\dots$ one by one, where $X_{p,i}$ are independent to each other and follow the same distribution of $X_p$. We can show that, using the approximation distributions in the computation will not lead to much error.
	
	\subsection{Adding Limitations to Strategy}
	For type $p$, if $c_p<\theta T$ ($\theta= \frac{\epsilon}{10n}$ as defined in the previous section), then there exists $e_p=2^{k_p},k_p\in Z$ such that $e_pc_p\in[\theta T, 2\theta T ]$. For convenience, if $c_p \ge \theta T$, we denote $e_p = 1$. We have the following restriction to the strategy.
	
	\begin{definition}[Restricted strategy]\label{restricted-strategy}
	A strategy is called restricted strategy, if for all type $p$, the total number of items of type $p$ we insert is always a multiple of $e_p$.
	\end{definition}
	
	If we know that for all type $p$, the total number of items of type $p$ is always a multiple of $e_p$, we hope that each time we use an item of type $p$, we will use $e_p$ of them together. This leads to the following definition.
	
	\begin{definition}[Block strategy]\label{block-strategy}
		A strategy is called block strategy, if we always insert a multiple of $e_p$ number of items of type $p$ together.
	\end{definition}
	
    The following theorem shows that, adding limitation to the strategy will not affect the optimal value too much.
	
	\begin{theorem}\label{bundle-err}
		Suppose the expected cost of the best block strategy is ${OPT}_{W}^{(b)}$, then
		$ OPT_{W} \leq  {OPT}_{W}^{(b)} \leq OPT_{W}+ \frac{\epsilon T}{5} $.
	\end{theorem}
	
	Because of the space limitation, we will present the proof sketch below. For the formal proof of Theorem \ref{bundle-err}, please see Appendix \ref{bundle-err-proof}.
	
	\begin{proof}[Proof sketch] The proof Theorem \ref{bundle-err} is divided into 2 parts. The first part shows that the optimal value for the original problem does not differ much from the optimal value with restricted strategy (see Definition \ref{restricted-strategy}), and the second part shows that the optimal value with restricted strategy is the same as the optimal value with block strategy (see Definition \ref{block-strategy}). The first part is simple since we can add some item after following the optimal strategy in the original problem. The second part follows from the intuition that if we must use an item in the future, it is good to use it right now.
	\end{proof}
	
	
	
	\subsection{Computing the Summation Distribution of Many Items of the Same Type}
	In the last part, we define block strategy by adding a constraint to the ordinary strategy. And we find the expected cost of the optimal block strategy is close to that of the optimal strategy.
	
	The block strategies conform to Assumption \ref{assm} in Section \ref{sec-a1}. If we know the distribution of the total weight of $e_p$ items of type $p$, we can compute the approximate optimal expected cost by Algorithm \ref{alg-1}. In this part, we give an algorithm which approximately computes the distribution of the total weight of $e_p$ items of type $p$.
	
	Due to the space limitation, we present our algorithm in this section, and we put the analysis of our algorithm into the appendix (see Appendix \ref{proof-main-thm}). To present our idea, we need the following definitions.
	\begin{definition}[Distribution Array]
		For a random variable $X$ with positive integer support, we use $X[i]$ to denote the probability that $X \ge i$, i.e. $X[i] = \text{Pr}\{X \ge i\}$, and we use an array $\text{Dist}(X):=(X[1],X[2],\dots,X[W])$ to denote the distribution. We call $\text{Dist}(X)$ the distribution array of variable $X$.
	\end{definition}
	
	\begin{remark}
		From the definition, we know that $\text{Dist}(X)$ is a non-increasing array. Besides, in the definition, $\text{Dist}(X)$ has only $W$ elements since we only care $X[i]$ when $i\leq W$.
	\end{remark}
	
	\begin{definition}
		For any non-increasing array $D=(D_1,D_2,\dots,D_W)$ of length $W$, if $D_1\leq 1$ and $D_W\geq 0$, there is a random variable $X$ such that $Dist(X)=D$. We say that $X$ is the variable corresponding to distribution array $D$, denoted by $Var(D):=X$.
	\end{definition}
	
	Suppose $\{Y_i\}_{i\geq 1}$ are identical independent random variables with distribution array $\text{Dist}(X_p)$. Let $S_i$ denote $\sum_{j=1}^{i}Y_j$ and $\text{Dist}(S_i)$ denote the corresponding distribution array. We want to compute the distribution array of $S_{e_p}$ and we have the following equations,
	\begin{align}
	\text{Pr}\{S_{2i}=w\} &=\sum_{j=1}^{w-1}\left(\text{Pr}\{S_{i}=j\}\cdot \text{Pr}
	\{S_{i}=w-j\}\right),\forall 1\le w \le W,\label{con-pdf}\\
	S_{2i}[w] &=\text{Pr}\{S_i \ge w\} + \sum_{j=1}^{w-1}\left(\text{Pr}\{S_{i}=j\}\cdot \text{Pr}\{S_i \ge w-j\}\right) \\
	&= S_i[w] + \sum_{j=1}^{w-1}\left( \left(S_{i}[j]-S_{i}[j+1]\right)\cdot S_{i}[w-j]\right).\label{con-ccdf}
	\end{align}
	Note that $S_{2i}$ can be computed from $S_{i}$, so we only need to compute $S_1,S_2,S_4\dots,S_{e_p}$ successively (recall that $e_p=2^k_p$ where $k_p\in Z$). Note that $S_1$ could be got from the oracle.
	
	However, computing the exact distribution of $S_{2^j}$ is slow (needs at least $poly(W)$ time), so we compute an approximate value of $S_i$. To introduce our method which approximately computes the distribution, we need the following definitions.
	
	\begin{definition}[$\eta$-Approximate Array]
		Given a positive real number $\eta$, for distribution array $A=(a_1,a_2,\dots,a_m)$, define $A'=(a'_1,a'_2,\dots,a'_m)$ as the $\eta$-approximate array of $A$, where for all $i\in[m]$,
		\begin{align*}
		a'_i=
		(1+\eta)^{\lceil\log_{1+\eta}a_i\rceil},& a_i>(1+\eta)^{-\zeta}
		.
		\end{align*}
	\end{definition}
	
	\begin{definition}[$(\zeta,\eta)$-Approximate Array]
		Given positive real numbers $\zeta,\eta$, for distribution array $A=(a_1,a_2,\dots,a_m)$, define $A'=(a'_1,a'_2,\dots,a'_m)$ as the $(\zeta,\eta)$-approximate array of $A$, where for all $i\in[m]$,
		\begin{align*}
		a'_i=
		\begin{cases}
		(1+\eta)^{\lceil\log_{1+\eta}a_i\rceil},& a_i>(1+\eta)^{-\zeta}  \\
		(1+\eta)^{-\zeta},& a_i\leq (1+\eta)^{-\zeta}
		\end{cases}
		.
		\end{align*}
	\end{definition}
	
	\begin{definition}[$(\zeta,\eta)$-Approximation]
		For random variable $X$, suppose distribution array $D$ is $(\zeta,\eta)$-approximate array of $\text{Dist}(X)$. Define $Var(D)$ as the $(\zeta,\eta)$-approximation of $X$.
	\end{definition}

    \begin{remark}
    The $(\zeta,\eta)$-Approximation of a random variable is still a random variable. And for any random variable $X$ with integer support in $[1,W]$, the $(\zeta,\eta)$-approximation of $X$ has at most $\lceil\zeta\rceil$ different possible values.
    \end{remark}
	
	Let $\eta=\frac{\epsilon}{10\log{W}}$ and $\zeta=\log_{1+\eta}{\frac{W}{\eta}}$, and our algorithm is shown as following: We first compute $(\zeta,\eta)$-approximation of $S_1$ which is denoted by $B_1$. Then for all $i\in \left[\lceil\log_{e_p}\rceil\right]$, we compute the distribution array of $B'_{2^{i}}$, which is the summation of independent $B_{2^{i-1}}$ and $B_{2^{i-1}}$. Then we compute $B_{2^{i}}$ which is the $(\zeta,\eta)$-approximation for $B'_{2^{i}}$. Finally, we can get $B_{e_p}$ which is an approximate random variable of $S_{e_p}$.
	
	When we compute the summation of $B_{2^{i-1}}$ and $B_{2^{i-1}}$, as there are at most $O(\zeta)$ different values in $\text{Dist}(B_{2^{i-1}})$, there are at most $O(\zeta)$ values $w$ such that $\text{Pr}\{B_{2^{i-1}}=w\}>0$. Based on the previous argument, we can enumerate $w_1$ and $w_2$ such that $\text{Pr}\{B_{2^{i-1}}=w_1\}>0$ and $\text{Pr}\{B_{2^{i-1}}=w_2\}>0$. In the end, we sort each $\text{Pr}\{B_{2^{i-1}}=w_1\}\cdot \text{Pr}\{B_{2^{i-1}}=w_2\}$ by the value $w_1+w_2$ and arrange them to get the distribution array $\text{Dist}(B'_{2^{i}})$. This shows that we can compute the approximate distribution in $O(\zeta^2 \log{\zeta})$ time.
	
	Formally, we have Algorithm \ref{alg-2} to compute $B_{e_p}$.
	
	
	\begin{algorithm}[!ht]
		\caption{Computing Approximate Distribution of $S_{e_p}$}\label{alg-2}
		\begin{algorithmic}[1] 
			\State Let $\eta=\frac{\epsilon}{10\log{W}}$ and $\zeta=\log_{1+\eta}{\frac{W}{\eta}}$.
			\State Let $B_1$ be the $(\zeta,\eta)$-approximation for $S_1$. Compute $\text{Dist}(B_1)$ according to the oracle
			\For{$i=1\rightarrow \log_{2}{e_p}$}
			\State Let $B'_{2^i}$ be the summation of $B_{2^{i-1}}$ and $B_{2^{i-1}}$. Compute $\text{Dist}(B'_{2^i})$.
			\State Compute $\text{Dist}(B_{2^i})$, which is the $(\zeta,\eta)$-approximation for $B_{2^i}$.
			\EndFor
			\State \Return $\text{Dist}(B_{e_p})$
		\end{algorithmic}
	\end{algorithm}
	
	
	Before we state the main theorem that bounds the approximation error of our algorithm, we combine the full procedure and get our final Algorithm \ref{alg-3} for \prob.
	
	
	\begin{algorithm}[!ht]
		\caption{Algorithm for \prob}\label{alg-3}
		\begin{algorithmic}[1] 
			\State For each type $p$, let $S^p_{e_p}$ be the summation of $e_p$ i.i.d. variables with distribution $X_p$. Compute approximate distribution $Y_p$ of $S^p_{e_p}$ by Algorithm \ref{alg-2}.
			\State For each item type $p$, construct new type $p'$ with expected cost $c_pe_p$ and weight distribution $Y_p$.
			\State Let $W$ and all the new types be the input, and get return value $\hat{V}$ of Algorithm \ref{alg-1}.
			\State \Return $\hat{V}$.
		\end{algorithmic}
	\end{algorithm}
	
	
	Then, we have our main theorem, which discusses the approximation error of Algorithm \ref{alg-3}.
	\begin{theorem}\label{main-thm}
		The output $\hat{V}$ of Algorithm \ref{alg-3} satisfies
		\[(1-\epsilon)OPT_W\leq \hat{V}\leq (1+\epsilon)OPT_W.\]
	\end{theorem}
	
	To prove this theorem, we first show $\text{Dist}(B_{e_p})$ in Algorithm \ref{alg-2} is approximation of $\text{Dist}(A_{e_p})$, by constructing another strategy $C_{e_p}$ which is strictly better than $B_{e_p}$ and the expected cost of $C_{e_p}$ is closed to the expected cost of $A_{e_p}$ (induction is used). Then we combine all the errors in Algorithm \ref{alg-3} and prove that Algorithm \ref{alg-3} is FPTAS of \prob. For details, please see Appendix \ref{proof-main-thm}.
	
	\subsection{Time Complexity}
	Our algorithm runs in $poly(n,\log W, \frac{1}{\epsilon})$. Combined with Theorem \ref{main-thm}, Algorithm \ref{alg-3} is an FPTAS for \prob. The theorem for the time complexity of Algorithm \ref{alg-3} is stated as follow,
	
	\begin{theorem}\label{time-complexity}
		Algorithm \ref{alg-3} runs in polynomial time and thus is an FPTAS for \prob. More specifically, Algorithm \ref{alg-3} has time complexity
		\[O\left( \frac{n\log^6 W}{\epsilon^3}+\frac{n^3\log W}{\epsilon^4} \right).\]
	\end{theorem}
	
	This theorem can be proved by recalling the parameters we have set, counting for the number of each operation, and expanding the parameters as $n,\log W$ and $\frac{1}{\epsilon}$. Please see Appendix \ref{proof-time-complexity} for the formal proof.

	\section{Conclusions and Further Work}
	We obtain the first FPTAS for \prob~in this paper.
	We focus on computing approximately the optimal value,
	but our algorithms and proofs immediately
	imply how to construct an approximate strategy in polynomial time.
	
	There are some other directions related to \prob~which are still open.
	It would be interesting to design a PTAS (or FPTAS) for the 0/1 stochastic minimization knapsack problem, the 0/1 stochastic (maximization) knapsack problem and the stochastic unbounded (maximization) knapsack problem.
	Hopefully, our techniques can be helpful in solving these problem.
	%
	%
	%

    \section*{Acknowledgement}

    The authors would like to thank Jian Li for several useful discussions and the help with polishing the paper.
The research is supported in part by the National Basic Research Program of China Grant 2015CB358700,
the National Natural Science Foundation of China Grant 61822203, 61772297, 61632016, 61761146003,
and a grant from Microsoft Research Asia

	\nocite{RN14}
	\bibliographystyle{splncs04}
	\bibliography{refx}
	%
	
	
	
	
	
	\newpage
	\appendix
	
	\section{Detailed Version of Algorithm \ref{alg-1}}\label{full-pseudo-code}
	In this section, we provide the detailed version of Algorithm \ref{alg-1}, which is shown below as Algorithm \ref{alg-1-detailed}.
	
	\begin{algorithm}[!ht]
		\caption{The Dynamic Program for Computing approximate answer for \prob~under Assumption \ref{assm}: Detailed Version}\label{alg-1-detailed}
		\begin{algorithmic}[1] 
			\State Let $\delta = \frac{\epsilon^2}{100n}$.
			\State $f_{0}\leftarrow 0$
			\For{$i=1\rightarrow \lceil\frac{1}{\delta}\rceil+1$} \Comment{The DP for computing $f_i$}
			\State Let $L=0,R=W$
			\While{$R>L$} \Comment{Binary search to guess $f_i$}
			\State $w\leftarrow \lfloor(L+R)/2\rfloor$
			\State $e\leftarrow 0$ \Comment{$e$ means whether $g_w$ is less or equal to $i\delta T$}
			\For{$j=1\rightarrow n$}
			\For{$m=0\rightarrow i-2$} \Comment{Compute the probability after rounding the weight}
			\State $P_m=Pr\left[X_j\in [w-f_{m+1},w-f_{m}) \right]$ \Comment by Binary search from oracle
			\EndFor
			\State $P_{i-1}=Pr\left[X_j\in [1,w-f_{i-1}) \right]$
			\State $g \leftarrow c_j+\sum_{m=0}^{i-1}P_m (m+1) \delta T$ \Comment Equation (\ref{g-to-f})
			\If{$g\leq i \delta T$}
			\State $e\leftarrow 1$
			\EndIf
			\EndFor
			\If{$e=1$}
			\State $L=w$
			\Else
			\State $R=w-1$
			\EndIf
			\EndWhile
			\State $f_i\leftarrow L$
			\If{$f_i\geq W$}
			\State \Return $\hat{V}:=i\delta T$
			\EndIf
			\EndFor
		\end{algorithmic}
	\end{algorithm}

    \section{Proof of Proposition \ref{positive-support-assumption}}\label{proof-positive-sup}

    \begin{proof}[Proof of Proposition \ref{positive-support-assumption}]
		
		Firstly recall that in the definition of \prob, $D_i$ has non-negative integer support. If we add an item of type $i$ and the realized weight $X_i = 0$, because there are infinite number of items of each type and the state does not change, from the dynamic program, we should also add the item of type $i$ until the realized weight is not $0$. We can construct another type $i'$ with distribution $D'_i$ and cost $c'_i$ to replace type $i$, where using one item of type $i'$ is equivalent to using several items of type $i$ until the realized weight of one item is positive. Then, $D'_i$ has positive integer support, and formally speaking, we have
		\[c'_i = \frac{c_i}{1-d_i(0)}, d'_i(t) = \frac{d_i(t)}{1-d_i(0)},\forall t>0,\]
		where $d'_i(t) = D'_i(t) - D'_i(t-1)$ (recall that $d_i(t) = D_i(t) - D_i(t-1)$).

        Then we can get
		\[D'_i(t) = \sum_{j=1}^t d_i(j) = \frac{D_i(t) - D_i(0)}{1-d_i(0)}.\]
		\qed
	\end{proof}

	\section{Proof of Lemma \ref{constapprox}}\label{proof-constapprox}
	\begin{proof}[Proof of Lemma \ref{constapprox}]
		Prove by induction. First, for all $-W+1 \le w\le 0$, $OPT_w=0$,
		\[b_{m}w\le OPT_w \le b_{m}(w+W).\]
		Suppose for all $w\le k,k\geq 0$, $b_{m}w\le OPT_w \le b_{m}(w+W)$. Then we have that
		\begin{align*}
		OPT_{k+1} =& \min_{i}\left(c_i + \mathbb E_{X_i} \left[OPT_{k+1-X_i}\right]\right) \\
		\le& c_m + \mathbb E_{X_m} \left[OPT_{k+1-X_m}\right] \\
		=& c_m + \sum_{w=1}^W d_m(w)OPT_{k+1-w} \\
		\le& c_m  + \sum_{w=1}^W d_m(w)b_m(k+1-w+W) \\
		=& c_m  + b_{m}(k+1+W) - b_{m}\mathbb E[X_m] \\
		=& b_{m}(k+1+W).
		\end{align*}
		\begin{align*}
		OPT_{k+1} =& \min_{i}\left(c_i + \mathbb E_{X_i} \left[OPT_{k+1-X_i}\right]\right) \\
		=& \min_{i}\left( c_i + \sum_{w=1}^W d_i(w)OPT_{k+1-w} \right) \\
		\ge& \min_{i}\left( c_i + \sum_{w=1}^W d_i(w)b_i(k+1-w) \right) \\
		=& \min_{i}\left(  c_i  + b_{i}(k+1) - b_{i}\mathbb E[X_i] \right) \\
		=& \min_{i}\left(  b_{i}(k+1)  \right)\\
		\ge& b_{m}(k+1).
		\end{align*}
		Then we arrange the terms, and we get
		\[b_{m}(k+1)\le OPT_{k+1} \le b_{m}(k+1+W).\]
		Above, we complete the proof by induction.\qed
	\end{proof}
	
	\section{Proof of Theorem \ref{dp-err}}\label{proof-dp-err}
	In this section, we will analyze the approximation error of Algorithm \ref{alg-1} and prove Theorem \ref{dp-err}. We will rely on two lemmas to prove Theorem \ref{dp-err}. Generally speaking, the first lemma shows that the real optimal value is upper bounded by the value computed in our algorithm, and the second lemma shows that under Assumption \ref{assm}, the error between the real optimal value is lower bounded by the difference between the value computed in our algorithm and another small value.
	
	Before proving the theorem, let's first recall Assumption \ref{assm}. Assumption \ref{assm} says that, for each type $i$, $c_i\ge \theta T$, where $\theta= \frac{\epsilon}{10n}$ and $c_i$ is the cost of type $i$. Then, we will also recall the variables and notations defined previously.
	
	We use $OPT_W$ denote the optimal value of \prob~and $\overline{OPT}_{W}$ denote the estimation of the optimal value in Section \ref{sec-a0}. We define $T=\frac{1}{4}\overline{OPT}_{W}$.
	
	Similar to the FPTAS of the ordinary knapsack problem, we define $\hat{f}_c=\max\{w|OPT_w\le c\}$, and compute an approximation for $\hat f$. However, the support of $\hat f$ is the set of real numbers. So we discretize $\hat f$ and only compute the approximation for $\hat f_{i\delta T}$ for all $i \le \lceil\frac{1}{\delta}\rceil + 1$, where $\delta = \frac{\epsilon^2}{100n}$. In our algorithm, we use dynamic programming to compute $f_i$, which is the approximation for $\hat f_{i\delta T}$.
	
	We also define $\hat g_w:=j\delta T$, for all $\hat f_{j-1}< w \leq \hat f_{j}$. Then $\hat g_w$ is the rounding up discretization value of $OPT_w$. In the algorithm, we use the following quantity to approximate $\hat f$. Given $f_0,\dots,f_{i-1}$, define $g_w:=j\delta T$ for all $f_{j-1}< w \leq f_{j},j\le i-1$, and define $g_w = i\delta T$ for all $w>f_{i-1}$. The ideas behind Algorithm \ref{alg-1} and the process to compute $f_i,g_w$ are shown in Section \ref{sec-a1}.
	
	Then, we will prove Theorem \ref{dp-err}. We first have the following lemmas.
	
	\begin{lemma}\label{cor-dp-1}
		
		For all $f_{i}$, we have $f_{i}\leq \hat{f}_{i\delta T}$, which means $  OPT_{f_{i}} \leq i\delta T $.
		
	\end{lemma}
	
	\begin{proof}[Proof of Lemma \ref{cor-dp-1}]
		
		We prove this by induction. The base case is true, which is just $f_{0}=0$ and $\hat{f}_{0}=0$. Now, assume the statement is true for $f_{j}$ where $j\leq i-1$. We prove the statement is also true for $f_{i}$.
		
		For $0\leq j<i$, $  OPT_{f_{j}} \leq j\delta T $. So for all $1\leq j\leq i-1$, for all $f_{j-1}< w\leq f_{j}$,
		$$OPT_{w} \leq OPT_{f_{j}} \leq j \delta T. $$
		
		We know $g_{w}=j\delta T$, so $OPT_{w}\leq g_{w}$ for all $0\leq w\leq f_{i-1}$.
		
		When we compute $f_i$, we define $g_w=i\delta T$ for all $w>f_{i-1}$, so $OPT_{w}\leq g_{w}$ for all $w\leq \hat{f}_{i\delta T}$.
		
		We know
		\[f_i=\max \left\{ w ~\Bigg{|}~
		\exists k, \left( c_k+\sum_{j=1}^{w-1} d_k(w-j) g_{j} \right)
		\leq i\delta T  \right\},\]
		and
		\[\hat{f}_{i\delta T}=\max\left\{w'~\Bigg{|}~\exists k,
		\left( c_k + \sum_{j=1}^{w'-1} d_k(w'-j)OPT_{j} \right) \le i\delta T\right\}.\]
		
		So, $f_i\leq \hat{f}_{i\delta T}$, which implies $OPT_{f_i}\leq i\delta T$.\qed
	\end{proof}
	
	\begin{lemma}\label{cor-dp2}
		
		For all $f_{i}$, $ OPT_{f_{i}+1} > i\delta T - \frac{i\delta^2}{\theta}T$.
		
	\end{lemma}
	
	\begin{proof}[Proof of Lemma \ref{cor-dp2}]
		
		We prove it by induction. The base case is true which is $f(1)>0$. Now assume the statement is true for $f_{j}$ where $j\leq i-1$. We show that the statement is also true for $f_{i}$.
		
		For $0\leq j<i$, $ OPT_{f_{j}+1} > j\delta T - \frac{j\delta^2}{\theta}T$. So for all $1\leq j\leq i$, for all $f_{j-1}< w\leq f_{j}$,
		\[OPT_{w} \geq OPT_{f_{j-1}+1} > (j-1)\delta T - \frac{(j-1)\delta^2}{\theta}T. \]
		
		We know $g_{w}=j\delta T$, where $f_{j-1}< w\leq f_{j}$, so
		\begin{align}\label{lim-g-opt}
		OPT_{w} \geq \left(1-\frac{\delta}{\theta}\right) g_{w} - \delta T,~~\forall f_{j-1}< w\leq f_{j}.
		\end{align}
		
		We know
		\[
		f_i=\max \left\{ w' ~\Bigg{|}~
		\exists k, \left( c_k+\sum_{j=1}^{w'-1} d_k(w'-j) g_{j} \right)
		\leq i\delta T  \right\}.
		\]
		
		Let $w^*=f_{i}+1$. Then for all $k$,
		\[ \left( c_k+\sum_{j=1}^{w^{*}-1} d_k(w^{*}-j) g_{j} \right) > i\delta T.  \]
		
		From \ref{lim-g-opt}, we know that for all $k$,
		\begin{align*}
		&~~~~\left( c_k+\sum_{j=1}^{w^{*}-1} d_k(w^{*}-j) OPT_{j} \right)  \\
		&\geq \left( c_k-\delta T +\left(1-\frac{\delta}{\theta}\right)\sum_{j=1}^{w^{*}-1} d_k(w^{*}-j) g_{j} \right) \\
		&\geq \left( c_k-\frac{\delta}{\theta} c_k +\left(1-\frac{\delta}{\theta}\right)\sum_{j=1}^{w^{*}-1} d_k(w^{*}-j) g_{j} \right) \\
		&\geq \left(1-\frac{\delta}{\theta}\right) \left( c_k + \sum_{j=1}^{w^{*}-1} d_k(w^{*}-j) g_{j} \right) \\
		&>\left(1-\frac{\delta}{\theta}\right)i\delta T    \\
		&=i\delta T - \frac{i\delta^2}{\theta}T.   \\
		\end{align*}
		It means
		\[OPT_{w^{*}}= \min_{k}\left\{c_k + \sum_{j=1}^W d_k(w^*-j)OPT_{j}\right\} >i\delta T - \frac{i\delta^2}{\theta}T.\]
		
		Then, we complete the proof by induction.\qed
	\end{proof}
	
	Given Lemma \ref{cor-dp-1} and Lemma \ref{cor-dp2}, we can prove the main theorem (Theorem \ref{dp-err}) in Section \ref{sec-a1}.
	
	\begin{proof}[Proof of Theorem \ref{dp-err}]
		First, with our algorithm, we have $\hat V = i\delta T$, where $f_{i-1} <W \le f_i$. Then because $OPT_w$ is increasing with respect to $w$, we know that
		\[OPT_{f_{i-1}+1} \le OPT_W \le OPT_{f_{i}}.\]
		Combining Lemma \ref{cor-dp-1} and Lemma \ref{cor-dp2}, we have
		\[OPT_{f_{i}} \leq i\delta T,~OPT_{f_{i-1}+1} > (i-1)\delta T - \frac{(i-1)\delta^2}{\theta}T = (i-1)\left(1-\frac{\epsilon}{10}\right)\delta T,\]
		which leads to
		\[(i-1)\left(1-\frac{\epsilon}{10}\right)\delta T \le OPT_W \le i\delta T.\]
		Recalling our definition of $T$, we have $T \le OPT_W$. Then we have $(i-1)\left(1-\frac{\epsilon}{10}\right) \ge \frac{1}{\delta}$, which implies $i \ge \frac{1}{\delta}$. Then from the fact that $\frac{i-1}{i} = 1 - \frac{1}{i} \ge 1-\delta$, we have
		\[(1-\delta)\left(1-\frac{\epsilon}{10}\right)i\delta T \le OPT_W \le i\delta T,\]
		which completes the proof of Theorem \ref{dp-err}.\qed
	\end{proof}

	\section{Proof of Theorem \ref{bundle-err}}\label{bundle-err-proof}
	
		Suppose the expected cost of the best restricted strategy is ${OPT}_{W}^{(r)}$. Obviously, we have $OPT_{W}\leq {OPT}_{W}^{(r)}$.
	
	\begin{lemma}\label{restrict-err}
		We have
		\[ {OPT}_{W}^{(r)} \leq OPT_{W}+ \frac{\epsilon T}{5}. \]
	\end{lemma}
	
	\begin{proof}
		
		In the best strategy with expected cost $OPT_{W}$, after the total weight of items in the knapsack reaches $W$, we can insert some extra items such that the total number of items of type $p$ in the knapsack is a multiple of $e_p$.
		By inserting extra items, the original strategy becomes a restricted strategy. And the cost of extra items of one type is at most $2\theta T$. So the total extra cost is at most $2n\theta T=\frac{\epsilon T}{5}$.
		\qed
	\end{proof}
	
	
	The following lemma shows, in a strategy, if we must insert a multiple of $e_p$ number of items of type $p$, it is equivalent to a strategy such that at each time we choose to insert an item of type $p$, we insert $e_p$ of them together.
	\begin{lemma}\label{change-order}
		In the process we insert items, if we must insert at least one more item of type $p$ (recall that the number of items of type $p$ we insert in the end must be multiple of $e_p$), Inserting one item of type $p$ right now is one best strategy.
	\end{lemma}
	
	\begin{proof}
		We first insert one item of type $p$ right now and pretend we didn't add this item. Then we insert items under the guidance of the original strategy until the original strategy tells us to insert the item of type $p$ into the knapsack. At present, we do not add it, and we regard the item of type $p$ we inserted before as the item we should insert at present.
		
		It is clear that the new strategy has the same expected cost as the original strategy.\qed
	\end{proof}
	
	The proof of Theorem \ref{bundle-err} follows directly from Lemma \ref{restrict-err} and Lemma \ref{change-order}.
	
	\section{Proof of Theorem \ref{main-thm}}\label{proof-main-thm}
	In this section, we will analyze the approximation error of Algorithm \ref{alg-3} and prove Theorem \ref{main-thm}. To begin with, we need some definitions to illustrate our points.
	\begin{definition}[Sub-strategy]
		In the strategy, we usually adaptively insert a series of items with the same type. We call this process sub-strategy.
		
		Let $\text{Weight}(s)$ denotes the total weight of items inserted in a sub-strategy $s$, noting that $\text{Weight}(s)$ is a random variable. To be convenient, we use $DistW(s)$ to denote $Dist(\text{Weight}(s))$.
		
		Define the expected cost of sub-strategy $s$ as the expected total cost of items inserted in a sub-strategy $s$.
	\end{definition}
	
	\begin{remark}
		We can view the sub-strategy as a series of functions $\{f_{i,p}:\mathbb Z^{i-1}\to [0,1]\}$, where $f_{i,p}$ denotes the choice function of the $i^{th}$ item in the sub-strategy. The $i-1$ parameters in $f_{i,p}$ denote the realized weights of the previous items of type $p$, and the output denote the probability to insert a new item of type $p$, otherwise, we do not insert a new item and the whole sub-strategy is stopped. As for the decision of the $i^{th}$ item, we randomly choose an action based on the output of $f_{i,p}$.
	\end{remark}
	For example, suppose an item of a certain type will weigh $2$ or $3$ with the same probability $0.5$, and the cost of one item is $1$. One sub-strategy (denoted by $s$) is: First insert one item. If it weighs 2, then we insert another and stop, otherwise, it will stop immediately. The expected cost of the sub-strategy is $1.5$, and the distribution array of $W(s)$ is $(1,1,1,0.5,0.25,0,0,\dots,0)$, since the total weight is $3$ with probability $0.5$, $4$ with probability $0.25$ and $5$ with probability $0.25$.
	\begin{definition}
		Given two distribution arrays $\text{Dist}(A)=(A[1],A[2],\dots,A[W])$ and $\text{Dist}(B)=(B[1],B[2],\dots,B[W])$, if for all $i$, $A[i]\geq B[i]$, then we use $\text{Dist}(A)\geq \text{Dist}(B)$ to denote this relation.
	\end{definition}
	\begin{definition}
		Given one distribution array $\text{Dist}(A)=(A[1],A[2],\dots,A[W])$ and one positive real number $\lambda$, define
		$$\lambda \cdot \text{Dist}(A):=(\min(1,\lambda A[1]),\min(1,\lambda A[2]),\dots,\min(1,\lambda A[W])).$$
	\end{definition}
	\begin{lemma}\label{jihe}
		Suppose that there is a sub-strategy $s$ with expected cost $c_s$ and distribution array $\text{DistW}(s)=(s[1],s[2],\dots,s[W])$. Then for any $\lambda>1$, there exists a sub-strategy $t$ with expected cost $\lambda c_s$, and $t[w]\geq \min(1,\lambda s[w])$ for all $w \in [1,W]$, i.e, $\text{DistW}(t)\geq \lambda\cdot \text{DistW}(s)$.
	\end{lemma}
	\begin{proof}
		
		Given a sub-strategy $s$ and the distribution array $\text{DistW}(s)=(s[1],s[2],\dots,s[W])$, supposing $w$ denote the realized total weight of the sub-strategy $s$, define random variable $r_s$ as the quality factor of $s$, where $r_s$ is uniformly distributed in $\left( s[w+1],s[w] \right]$ given $w$. (Here, we denote $s[W+1]=0$.)
		
		So, the quality factor $r_s$ is uniformly distributed in $(0,1]$ without knowing $w$. Further more, the larger $w$ is, the smaller $r_s$ is.
		
		We build a sub-strategy $t$ given as follow:
		
		We first execute sub-strategy $s$. If the quality factor $r_s$ satisfies $r_s \leq \frac{1}{\lambda}$, then stop, otherwise we repeat executing sub-strategy $s$ until the quality factor satisfies $r_s\leq \frac{1}{\lambda}$.
		
		Let $c_s$ and $c_t$ denote the expected cost of $s$ and $t$ respectively. They should satisfy $c_t=c_s+(1-\frac{1}{\lambda})c_t$, which implies $c_t=\lambda c_s$.
		
		Then we prove that for all $w$, we have $t[w]\geq \min(1,\lambda s[w])$. If $s[w] > \frac{1}{\lambda}$, because we construct sub-strategy $t$ such that we will get a realized $s$ with the quality factor $\leq \frac{1}{\lambda}$, then in the last time we execute $s$, total weight in $s$ must reach $w$, so $t[w]= 1$. If $s[w]\leq\frac{1}{\lambda}$, then we also consider the last time we execute $s$. The quality factor is at most $\frac{1}{\lambda}$ in this case, so the probability that the total weight is greater than $w$ given the quality factor is at most $\frac{1}{\lambda}$ is $\lambda s[w]$.\qed
	\end{proof}
	\begin{lemma}\label{round-error}
		Given a sub-strategy $s$ whose expected cost is $c_s$, suppose the $\eta$-approximate array of $DistW(s)$ is $D'$, and the $(\zeta,\eta)$-approximate array of $DistW(s)$ is $D''$. Then there is a sub-strategy $t_1$ whose expected cost is $(1+\eta)c_s$ and $\text{DistW}(t_1)\geq D'$. And there is a sub-strategy $t_2$ whose expected cost is at most $(1+\eta)^2 c_s$ and $\text{DistW}(t_2)\geq D''$.
	\end{lemma}
	\begin{proof}
		First, $D'\leq (1+\eta)\text{DistW}(s)$ from the definition of $\eta$-approximation. Then from Lemma \ref{jihe}, sub-strategy $t_1$ exists.
		
		Then we construct sub-strategy $t_2$. The sub-strategy is: Execute sub-strategy $t_1$ first. If its quality factor is larger than $1-(1+\eta)^{-\zeta}$, insert $W$ more items of type $p$ to ensure that the total weight reaches $W$. Then $\text{DistW}(t_2)\geq D''$ simply follows from the definition of sub-strategy $t_2$. The expected cost of $t_2$ is
		\[ c_{t_2}=(1+\eta) c_s + (1+\eta)^{-\zeta}W c_s\leq (1+\eta)^2 c_s. \]
		\qed
	\end{proof}
	\begin{lemma}
		Given two variables $X$ and $Y$ which satisfy $\text{Dist}(X)\geq \text{Dist}(Y)$, and another variable $Z$ with distribution array $\text{Dist}(Z)$, then $\text{Dist}(X+Z)\geq \text{Dist}(Y+Z)$.
	\end{lemma}
	\begin{proof}
		For all $w$, similar to Equation \ref{con-ccdf}, we have
		\begin{align*}
		(X+Z)[w] =& Z[w] + \sum_{j=1}^{w-1}\left( \left(Z[j]-Z[j+1]\right)\cdot X[w-j] \right)\\
		\geq& Z[w] + \sum_{j=1}^{w-1}\left( \left(Z[j]-Z[j+1]\right)\cdot Y[w-j] \right) \\
		=& (Y+Z)[w].
		\end{align*}
		\qed
	\end{proof}
	\begin{corollary}\label{order-convo}
		Given two variables $X_1, X_2$ with the same distribution $\text{Dist}(X)$ and two variables $Y_1,Y_2$ with the same distribution $\text{Dist}(Y)$, if $\text{Dist}(X)\geq \text{Dist}(Y)$, then $\text{Dist}(X_1+X_2)\geq \text{Dist}(X_1+Y_1) \geq \text{Dist}(Y_1+Y_2)$.
	\end{corollary}
	Define sub-strategy $A_{2^i}$ which always use $2^i$ items of type $p$, and define random variable $S_{2^i}$ as the total weight of $2^i$ items of type $p$. Then we have $\text{DistW}(A_{2^i})=\text{Dist}(S_{2^i})$. Recalling the variable $B_{2^i}$ in our algorithm, when computing $B_{2^i}$, we always round the value larger (from the definition of $(\zeta,\eta)$-approximation). So, from the induction argument, $\text{Dist}(B_{2^i})\geq \text{DistW}(A_{2^i})$ for all $i$. We restate the argument into the following corollary.
	
	\begin{corollary}\label{bounded-below}
		For all $i$ such that $2^i \le e_p$,
		\[\text{Dist}(B_{2^i})\geq \text{DistW}(A_{2^i}) = \text{Dist}(S_{2^i}).\]
	\end{corollary}
	
	Now we want to claim that $B_{2^i}$ does not deviate too much from $A_{2^i}$. Let $c_p$ denote the cost of one item of type $p$. We have the following lemma.
	\begin{lemma}\label{bounded-above}
		For all $i \le \log_2 e_p$, there is a sub-strategy $C_{2^i}$, whose expected cost is at most $(1+\eta)^{2(i+1)} 2^i c_p$, and $\text{DistW}(C_{2^i})\geq \text{Dist}(B_{2^i})$.
	\end{lemma}
	\begin{proof}
		$\text{Dist}(B_{2^0}) = \text{Dist}(B_1)$ is the $(\zeta,\eta)$-approximate array of $\text{DistW}(A_{2^0})$, from Lemma \ref{round-error}, $C_{2^0}$ exists.
		
		Now, assume that for all $i<k$, $C_{2^i}$ exists, and we prove $C_{2^k}$ exists. Recall that $B'_{2^k}$ is the summation of $B_{2^{k-1}}$ and $B_{2^{k-1}}$, and we first construct sub-strategy $C'_{2^{k}}$ which simply executes $C_{2^{k-1}}$ twice. Then the expected cost of $C'_{2^{k}}$ is $(1+\eta)^{2k} 2^k c_p$, and from Corollary \ref{order-convo}, we know $\text{DistW}(C'_{2^{k-1}})\geq \text{Dist}(B'_{2^{k}})$. Recall $B_{2^{k}}$ is the $(\zeta,\eta)$-approximation of $B'_{2^{k}}$. From Lemma \ref{round-error}, we know there exists $C_{2^{k}}$ whose expected cost is $(1+\eta)^{2(k+1)} 2^k c_p$, and $\text{DistW}(C_{2^{k}})\geq \text{Dist}(B_{2^{k}})$. Then we complete the proof by induction.\qed
	\end{proof}
	
	We can view the sub-strategy as a new type of item with cost following a distribution, and we know the expected cost of the sub-strategy. From the previous argument, we know that the distribution array of $B_{2^i}$ is bounded by $2$ other distribution array $A_{2^i}$ and $C_{2^i}$, whose expected costs are `close to each other.' We show that our computation leads to small error compared with the optimal value.
	\begin{proof}[Proof of Theorem \ref{main-thm}]
		First, $OPT_W$ is the optimal value for \prob, and $OPT_W^{*}$ is the optimal value such that for each type $p$ with $c_p < \theta T$, we use a multiple of $e_p$ number of items of type $p$ together. Then from Theorem \ref{bundle-err}, we have
		\[OPT_W \le {OPT}_{W}^{*} \leq OPT_{W}+ \frac{\epsilon T}{5}.\]
		Let $OPT^{**}_W$ denote the optimal value with items $\{(c_pe_p, Y_p)\}_{p\le n}$, where $c_pe_p$ is the expected cost of the item and $\text{Dist}(Y_p)$ is the weight distribution array. We next show that
		\[\frac{1}{(1+\eta)^{2\log_2 W +2}}OPT_W \le OPT^{**}_W \le OPT^*_W.\]
		We first show that $OPT^{**}_W \le OPT^*_W$. Note that $OPT^*_W$ is the optimal value with items $\{(c_pe_p, S^p_{e_p})\}_{p\le n}$, and from Corollary \ref{bounded-below}, we have $\text{Dist}(Y_p) \ge \text{Dist}(S^p_{e_p})$. First, it is obvious that $OPT^*_w$ and $OPT^{**}_w$ are non-decreasing with respect to $w$. Then let $OPT^*_w = OPT^{**}_w = 1,\forall w\le 0$, and suppose for all $w\le k$, $OPT^{**}_w \le OPT^*_w$. Then for $w = k+1$,
		\begin{align*}
		OPT^*_w =& \min_{i}\left\{c_ie_i + \sum_{j=1}^{W-1} (S^i_{e_i}[j] - S^i_{e_i}[j+1])OPT^*_{w-j} + S^i_{e_i}[W]OPT^*_{w-W} \right\} \\
		\ge& \min_{i}\left\{c_ie_i + \sum_{j=1}^{W-1} (S^i_{e_i}[j] - S^i_{e_i}[j+1])OPT^{**}_{w-j} + S^i_{e_i}[W]OPT^{**}_{w-W}\right\} \\
		=& \min_{i}\left\{c_ie_i + S^i_{e_i}[1]OPT^{**}_{w-1} + \sum_{j=2}^{W} S^i_{e_i}[j](OPT^{**}_{w-j} - OPT^{**}_{w-j+1})\right\} \\
		\ge& \min_{i}\left\{c_ie_i + Y_i[1]OPT^{**}_{w-1} + \sum_{j=2}^{W} Y_i[j](OPT^{**}_{w-j} - OPT^{**}_{w-j+1})\right\} \\
		=& \min_{i}\left\{c_ie_i + \sum_{j=1}^{W-1} (Y_i[j] - Y_i[j+1])OPT^{**}_{w-j} + Y_i[W]OPT^{**}_{w-W}\right\} \\
		=& OPT^{**}_w,
		\end{align*}
		where we use the Abel transformation, the monotonicity of $OPT^{**}_w$, the fact that $\text{Dist}(Y_p) \ge \text{Dist}(S^p_{e_p})$, and the fact that $Y_p[1] = S^p_{e_p}[1] = 1$ (all items have positive integer support).
		
		Then we prove that $OPT_W/(1+\eta)^{2\log_2 W +2} \le OPT^{**}_W$. First from Lemma \ref{bounded-above}, we know that there exists sub-strategies $\{(c'_p,Z_p)\}_{p\le n}$ such that  $c'_p \le (1+\eta)^{2\log_2 e_p + 2}e_pc_p$ and $\text{DistW}(Z_p) \ge \text{Dist}(Y_p)$, where $Y_p$ is the output distribution by Algorithm \ref{alg-2}. Note that $e_p \le W$, otherwise we have
		\[e_pc_p > Wc_p \ge OPT_W \ge 2\theta T \ge e_pc_p,\]
		which is impossible. Then we have $c'_p \le (1+\eta)^{2\log_2 W + 2}e_pc_p$. Note that a sub-strategy can be viewed as a type of item with cost following a distribution, and let $OPT'_w$ denote the optimal value with total weight $w$ and types of items $\{(c'_p,Z_p)\}_{p\le n}$. Without loss of generality, define $OPT'_w = 1,\forall w\le 0$. Then first, we have $OPT^{**}_w \ge OPT'_w / (1+\eta)^{2\log_2 W +2} ,\forall w\le 0$. Suppose for all $w\le k$, $OPT^{**}_w \ge OPT'_w / (1+\eta)^{2\log_2 W +2}$. Then for $w = k+1$,
		\begin{align*}
		OPT^{**}_w =& \min_{i}\left\{c_ie_i + \sum_{j=1}^{W-1} (Y_i[j] - Y_i[j+1])OPT^{**}_{w-j} + Y_i[W]OPT^{**}_{w-W}\right\} \\
		\ge& \min_{i}\Bigg\{c_ie_i + \frac{1}{(1+\eta)^{2\log_2 W +2}}\sum_{j=1}^{W-1} (Y_i[j] - Y_i[j+1])OPT'_{w-j} \\
        &\quad\quad + \frac{1}{(1+\eta)^{2\log_2 W +2}}Y_i[W]OPT'_{w-W}\Bigg \} \\
		\ge& \frac{1}{(1+\eta)^{2\log_2 W +2}}\min_{i}\Bigg\{ c'_i + Y_i[1]OPT'_{w-1} \\&\quad\quad+ \sum_{j=2}^{W} Y_i[j](OPT'_{w-j} - OPT'_{w-j+1})\Bigg\} \\
		\ge& \frac{1}{(1+\eta)^{2\log_2 W +2}}\min_{i}\Bigg\{c'_i + Z_i[1]OPT'_{w-1} \\&\quad\quad+ \sum_{j=2}^{W} Z_i[j](OPT'_{w-j} - OPT'_{w-j+1})\Bigg\} \\
		=& \frac{1}{(1+\eta)^{2\log_2 W +2}}\min_{i}\Bigg\{c'_i \\&\quad\quad+ \sum_{j=1}^{W-1} (Z_i[j] - Z_i[j+1])OPT'_{w-j} + Y_i[W]OPT'_{w-W}\Bigg\} \\
		=& \frac{1}{(1+\eta)^{2\log_2 W +2}}OPT'_w,
		\end{align*}
		where we use the Abel transformation, the monotonicity of $OPT'_w$, the fact that $\text{DistW}(Z_p) \ge \text{Dist}(Y_p)$, the fact that $Y_p[1] = Z_p[1] = 1$(all items have positive integer support) and the induction assumption. Then we prove $OPT'_W/(1+\eta)^{2\log_2 W +2} \le OPT^{**}_W$ by induction. Also notice that $OPT'_W$ is the optimal value using the sub-strategies $\{(c'_p,Z_p)\}_{p\le n}$, and each sub-strategy can be constructed through the original type of items, so it is clear that $OPT_W \le OPT'_W$. Then we have proved that
		\begin{align}\label{bound-appro-dist}
		\frac{1}{(1+\eta)^{2\log_2 W +2}}OPT_W \le OPT^{**}_W \le OPT^*_W.
		\end{align}
		
		Finally, from Theorem \ref{dp-err}, we can get
		
		\[(1-\delta)(1-\frac{\epsilon}{10})\hat{V} \leq OPT^{**}_{W} \leq \hat{V}.\]
		
		According to Theorem \ref{bundle-err} and Equation \ref{bound-appro-dist}, we can get
		\begin{align*}
		\frac{1}{(1+\eta)^{2\log_2 W +2}}OPT_W \leq &\hat{V} \leq  (1-\delta)^{-1}(1-\frac{\epsilon}{10})^{-1} \left( OPT+\frac{\epsilon T}{5} \right).
		\end{align*}
		
		So we have
		
		\begin{align*}
		(1-\epsilon)OPT_W \leq &\hat{V} \leq  (1+\epsilon)OPT_W.    \\
		\end{align*}
		\qed
	\end{proof}
	
	\section{Proof of Theorem \ref{time-complexity}}[Proof of Theorem \ref{time-complexity}]\label{proof-time-complexity}
	\begin{proof}
		Recall the value of parameters we select:
		\[\delta=\frac{\epsilon^2}{100n}\sim O\left(\frac{\epsilon^2}{n}\right),\]
		\[\eta=\frac{\epsilon}{10\log{W}}\sim O\left(\frac{\epsilon}{\log W}\right),\]
		and
		\[\zeta = \log_{1+\eta}{\frac{W}{\eta}}\sim O\left(\frac{\log^2 W}{\epsilon}+\frac{\log W\log{\frac{1}{\epsilon}}}{\epsilon}\right).\]
		
		Without loss of generality, we assume $\frac{1}{\epsilon}$ is $o(W)$. Otherwise, we can just use the pseudo-polynomial dynamic programming algorithm to compute the optimal value. Then,
		\[\zeta\sim O\left(\frac{\log^2 W}{\epsilon}\right).\]
		
		To compute the distribution of the summation of some i.i.d variables, for each type $p$, we compute approximate distribution array of $S_1,S_2,S_4,S_8,\dots,S_{e_p}$ where $e_p$ is $O(W)$, which means we need to compute $O(\log{W})$ distribution arrays. We compute the $(\zeta,\eta)$-approximate array, and get $(\zeta,\eta)$-approximate array of $S_1$ from the oracle by binary search which can be computed in $O(\zeta\log \zeta)$ time. And we compute $(\zeta,\eta)$-approximate array of $S_2,S_4,S_8,\dots,S_{e_p}$ by convolution in $O(\zeta^2\log\zeta)$ time. We also need to compute the approximate distribution array of at most $n$ types of items, so the total time is $O\left(n\zeta^2\log \zeta\log W\right)$
		
		In the dynamic programming, we compute $f_x$ where $x$ takes $\frac{1}{\delta}$ different values. When computing $f_x$, we need to use binary search on a value $y$, so we compute the approximate value of $OPT_y$ in $O(\frac{1}{\delta})$ time. We also need to enumerate which item type to select, so the total time is $O\left(n\frac{1}{\delta^2}\log{W}\right)$.
		
		Besides, we need to compute an approximate value $T$ by our constant factor approximation algorithm, which takes $O(n\log{W})$ time.
		
		In conclusion, the total time complexity is
		$$O\left( \frac{n\log^6 W}{\epsilon^3}+\frac{n^3\log W}{\epsilon^4} \right).$$
		
		So, Algorithm \ref{alg-3} runs in polynomial time.\qed
	\end{proof}

\end{document}